\documentclass[10pt, conference, compsocconf]{IEEEtran}
\IEEEoverridecommandlockouts

\usepackage{graphicx}

\ifCLASSINFOpdf
\else
\fi

\usepackage{amsmath}
\usepackage{amsthm}

\theoremstyle{definition}
\newtheorem{definition}{Definition}[section]
 
\newtheorem{theorem}{Theorem}[section]
\newtheorem{lemma}[theorem]{Lemma}

\usepackage{tabularx}
\usepackage{amsfonts}

\usepackage{algorithm}
\usepackage[noend]{algpseudocode}

\algnewcommand\algorithmicswitch{\textbf{switch}}
\algnewcommand\algorithmiccase{\textbf{case}}
\algnewcommand\algorithmicassert{\texttt{assert}}
\algnewcommand\Assert[1]{\State \algorithmicassert(#1)}%
\algdef{SE}[SWITCH]{Switch}{EndSwitch}[1]{\algorithmicswitch\ #1\ \algorithmicdo}{\algorithmicend\ \algorithmicswitch}%
\algdef{SE}[CASE]{Case}{EndCase}[1]{\algorithmiccase\ #1}{\algorithmicend\ \algorithmiccase}%
\algtext*{EndSwitch}%
\algtext*{EndCase}%

\begin{document}
%
\title{Collaborative Computation in Self-Organizing Particle Systems}


\author{\IEEEauthorblockN{Alexandra Porter\textsuperscript{*}}
\IEEEauthorblockA{Computer Science\\
Stanford University\\
Stanford, U.S.A\\
amporter@stanford.edu}

\thanks{\textsuperscript{*}This work was supported in part by NSF under Awards CCF-1353089 and CCF-1422603, and
matching NSF REU awards;  this work was conducted while the first author was an undergraduate student at ASU}
\and
\IEEEauthorblockN{Andr\'ea W.\ Richa\textsuperscript{*}}
\IEEEauthorblockA{Computer Science, CIDSE\\
Arizona State University\\
Tempe, U.S.A\\
aricha@asu.edu}

}


%


\newcommand{\Geqt}{\ensuremath{G_\text{eqt}}}
\newcommand{\Veqt}{\ensuremath{V_\text{eqt}}}
\newcommand{\Eeqt}{\ensuremath{E_\text{eqt}}}
\newcommand{\E}{\mbox{E}}
\newcommand {\OPT}{\ensuremath{\text{OPT}}}
\newtheorem{Claim}{\it Claim}

\maketitle

\begin{abstract}
   Many forms of programmable matter have been proposed for various tasks. We use an abstract model of self-organizing particle systems for programmable matter which could be used for a variety of applications, including smart paint and coating materials for engineering or programmable cells for medical uses. Previous research using this model has focused on shape formation and other spatial configuration problems (e.g., coating and compression). In this work we study foundational computational tasks that exceed the capabilities of the individual constant size memory of a particle, such as implementing a counter and matrix-vector multiplication. These tasks represent new ways to use these self-organizing systems, which, in conjunction with previous shape and configuration work, make the systems useful for a wider variety of tasks. They can also leverage the distributed and dynamic nature of the self-organizing system to be more efficient and adaptable than on traditional linear computing hardware. Finally, we demonstrate applications of similar types of computations with self-organizing systems to image processing, with implementations of image color transformation and edge detection algorithms.
\end{abstract}

\begin{IEEEkeywords}
self-organizing systems; programmable matter; distributed algorithms; binary counter; matrix multiplication
\end{IEEEkeywords}

%
\IEEEpeerreviewmaketitle

\section{Introduction}
The concept of \textit{programmable matter} was first defined by Toffoli and Margolus as a computing medium which can be used dynamically and in arbitrary amounts, controlled by both internal and external events~\cite{Toffoli1991263}.
Examples of programmable matter exist in nature, such as proteins closing wounds, bacteria building colonies, and the construction of coral reefs. These examples indicate potential applications of programmable matter, such as smart paint or coating materials for engineering, programmable cells for medical purposes, or adaptable and recyclable building blocks for everyday objects. These applications require tasks for which programmable matter is uniquely capable, such as shape formation and coating. However, they also require computations resembling those done by traditional computers to process information and make decisions. Work so far using the geometric amoebot model for self-organizing particle systems has focused on spatial configuration, including demonstrating efficient programmable matter algorithms for shape formation, coating, and compression (e.g., \cite{DBLP:conf/spaa/DerakhshandehGR16,DBLP:conf/dna/DerakhshandehGP16,DBLP:conf/dna/DerakhshandehGS15,DBLP:conf/podc/CannonDRR16}).

We introduce solutions using the amoebot model for basic computational tasks exceeding the capabilities of a single particle, including counting or number storage, and matrix-vector multiplication. Basic constructions for computational tasks can then be combined to solve more complex problems. However, many problems are structured in a way that makes a more specialized approach much more efficient. 
 We describe and analyze a binary counter algorithm and a matrix-vector multiplication algorithm using the amoebot model. We also discuss applications of the matrix-vector multiplication to image processing tasks, including color transformations and edge detection.

\subsection{Amoebot Model} \label{sec:model}

In the amoebot model, we represent the particle system as a subset of an infinite, undirected graph $G = (V,E)$, where $V$ is the set of all possible positions a particle can occupy, and $E$ is the set of all possible transitions between positions in $V$ \cite{DBLP:conf/dna/DerakhshandehGS15}.  In the {\em geometric amoebot model} we impose an underlying geometric structure for $G$ in the form of the \emph{equilateral triangular grid}, as shown in Figure~\ref{modelfig}a. Each particle occupies either a single node (i.e., it is {\em contracted}) or a pair of two adjacent nodes (i.e., it is {\em expanded}) on the graph, and each node can be occupied by at most one particle at any point in time, as shown in Figure~\ref{modelfig}b. Two distinct particles occupying adjacent nodes are connected by a {\em bond} and we refer to such particles as {\em neighbors}. The bonds ensure the particle system forms a connected structure and are used for exchanging information.

\begin{figure}
	\includegraphics[width=0.4\textwidth]{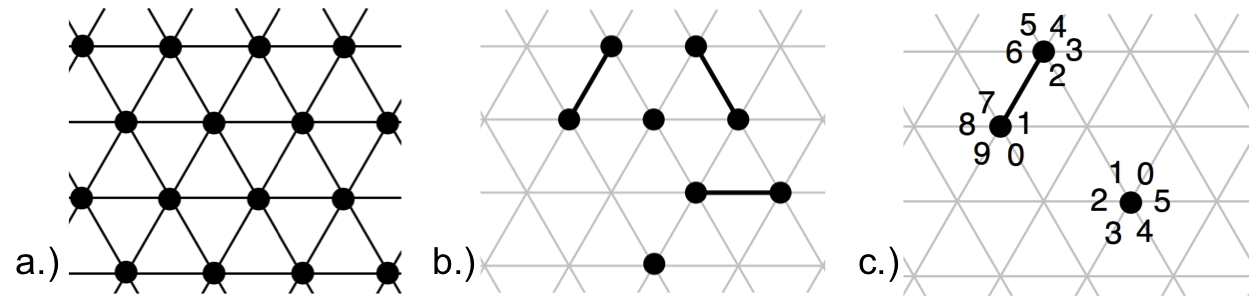}
			\caption{ (a) shows a section of $G$, where nodes of $G$ are shown as black circles.  (b) shows five particles on $G$; the underlying graph $G$ is depicted as a gray mesh; a contracted particle is depicted as a single black circle and an expanded particle is depicted as two black circles connected by an edge.  (c) depicts labeling of bonds for an expanded particle and a contracted particle.}
	\label{modelfig}
\end{figure}

 A particle at position $v \in V$ is {\em anonymous}, meaning it has no globally unique identifier. A particle uses a labeling function $\ell : E(v) \rightarrow \mathbb{N}$ to identify its neighbors and differentiate its possible bonds. Possible labelings for two nodes are shown in Figure~\ref{modelfig}c. Particles may communicate with their neighbors by reading and writing to their shared constant sized memory, which can equivalently be considered as the ability to pass a limited number of bounded-size tokens to adjacent particles. 

Particles move by asynchronously executing a series of \emph{expansions} and \emph{contractions}. If a particle occupies only one node, it is contracted and can expand to an unoccupied adjacent node. An expanded particle can then contract to occupy only one of the two nodes it occupied while expanded.

We assume a compass-free model, meaning there is no global sense of orientation shared by the particles, and we assume that the particles do not share any underlying coordinate system in $G$. In the case of the triangular grid, each particle $p$ fixes an arbitrary head direction, which specifies an adjacent edge $e_{head}$ to $p$. As shown in Figure~\ref{modelfig}, we assume particles have shared \emph{chirality} (sense of clockwise direction) and  so they can label their ports in a consistent direction (note that in the presence of gravity, chirality follows naturally). We assume they label their head and tail ports from $0$ to $5$, starting in an arbitrary direction.

 We assume an asynchronous, concurrent system of particles, where conflicts of movement (e.g., two particles trying to expand into the same empty node location) or shared memory (e.g., two adjacent particles trying to write concurrently onto their shared memory) are resolved arbitrarily so that at most one of the particles involved in the conflict ``wins''. Thus we can rely on the seminal results for the classical asynchronous model in distributed computing (see, e.g.,~\cite{lynch1996distributed}) that state that any asynchronous execution of the system, where conflicts are resolved arbitrarily, produces an equivalent outcome as a sequence of \emph{atomic particle activations}. Hence, we can assume, without loss of generality, that at most one particle is active at any point in time.
 Under this model, we define:
\begin{definition}\label{def:asynchround}
An \emph{asynchronous round} is given by the elapsed time until each particle has been activated at least once.
\end{definition}
In our context, when a particle is activated it can perform an arbitrary bounded amount of computation using its local memory and the shared memory of its neighbors, and at most one movement.

\subsection{Related Work}
There are a number of existing solutions for programmable matter, which can be categorized as active and passive systems. In passive systems, the computational units have no ability to control their motion, so they move and bond only based on their structure and environmental conditions. Passive systems include DNA computing and tile assembly models, in which computation occurs as a result of tiles bonding together in ways controlled by the tile attributes \cite{adleman1994molecular,doty2012theory,patitz2014introduction,woods2015intrinsic}.  Work on tile assembly considers computational problems similar to those we study, including demonstration of a binary counter \cite{rothemund2000program}.  However, the specifications of those systems (passive motion, unlimited supply of tiles of any type, etc.) differ considerably from ours. Active systems consist of computational units that control their actions, motions, and communications to accomplish specific tasks. Applications of active systems, including shape formation, coating, and compression, have been explored using robotic implementations  \cite{arbuckle2010self,das2010computational,kernbach2013handbook,woods2013active}. These applications have also been explored using the amoebot model, which is an active system \cite{DBLP:conf/spaa/DerakhshandehGR16,DBLP:conf/dna/DerakhshandehGP16,DBLP:conf/dna/DerakhshandehGS15,DBLP:conf/podc/CannonDRR16,DBLP:conf/nanocom/DerakhshandehGR15}. 

Classical algorithms for distributed matrix multiplication include Fox's \cite{fox1987matrix} and Cannon's \cite{cannon1969cellular}. These divide matrices into consecutive blocks to perform multiplication. More recent algorithms, including the Scalable Universal Matrix Multiplication Algorithm (SUMMA) \cite{van1997summa} and Distribution-Independent Matrix Multiplication Algorithm (DIMMA) \cite{choi1997new}, further reduce the number of necessary operations. In SUMMA, the matrix is divided into rows and columns of blocks, and values are then broadcast down columns and across rows. DIMMA improves on this by adding pipelining to communication and taking advantage of a Least-Common Multiple strategy to reduce computation requirements. Our simpler algorithm for matrix-vector multiplication broadcasts values down columns of the matrix in a way similar to how values are broadcast in SUMMA and DIMMA.

In the field of computer vision and image processing, matrix multiplication is used to apply operators for fundamental tasks including determining gradient (i.e.~\cite{sobel1990isotropic}) and measuring color invariants such as luminance~\cite{geusebroek2001color}. Basic color transformations operators, such as adjustments to brightness, saturation, and hue are also often used in image editing~\cite{haeberli1993matrix}.

An application of these image operators is edge detection, which is an important problem due to its applications in feature extraction and recognition. The edge detection algorithm introduced by Canny uses a series of steps including smoothing, filtering, and thresholding to extract edges from an image \cite{canny1986computational}. Research has been done into how to implement this method efficiently, including a distributed GPU implementation \cite{ogawa2010efficient} and other real-time methods using field-programmable gate arrays  \cite{gentsos2010real} and adaptive thresholding for improved accuracy \cite{neoh2004adaptive}.

Color transformations have also been combined with gradient operators (including those used in edge detection) for assessing feature similarity~\cite{zhang2011fsim} and quality of images~\cite{bae2016novel}. Other applications of color transformations include feature detection and analysis, particularly for human faces~\cite{plataniotis2013color}.

\subsection{Our Contributions}
We address the very basic and general problems of \emph{counting} and \emph{matrix-vector multiplication}. We describe the image processing applications of \emph{edge detection} and \emph{color transformations}, as examples of applications that can use and benefit from our matrix-vector and matrix-matrix multiplication setup and algorithms.  We assume each instance of these problems is fed into our particle system as a sequence of values passed through a seed particle. Results are stored distributed across the system, and can be output by each particle individually or passed to the seed to output the result as a data stream.

We present an algorithm for a basic {\em binary counter} using the amoebot model, and show that it counts to a value $v$ in {\em $O(v)$ asynchronous rounds}. We also present a two-part algorithm for {\em matrix-vector multiplication} using the amoebot model. The first part of the algorithm is to self-organize particles to set up the input matrix and vector and the resulting vector entries. The second part of the algorithm distributedly performs the actual multiplication (note that these two algorithmic components run concurrently and there is no need for synchronization). Let $h$ and $w$ denote the number of rows and columns of the matrix (unknown to the set of particles). We show that the {\em number of asynchronous rounds} it takes to {\em set up} the matrix and vector entries is $O(hw)$ and the {\em number of rounds required for matrix-vector multiplication} is $O(h+w)$. Extending this result by executing a sequence of matrix-vector multiplications, the {\em number of rounds required for matrix-matrix multiplication} is $O(y(h+w))$ with a second matrix of height $w$ and width $y$.  We describe and analyze a simple implementation of {\em Canny edge detection} in image processing, which utilizes the setup algorithm introduced for matrix-vector multiplication. We show this implementation requires {\em $O(1)$ rounds to complete edge detection} after the $O(hw)$ setup is completed (again no synchronization between these two algorithmic phases is needed). Finally, we describe the image processing application of \emph{color transformation}, which is setup in the same way with $O(hw)$ rounds and then requires $O(y(h+w))$ rounds for multiplication. We provide also experimental results on actual implementations of the Canny edge detection algorithm and the color transformation algorithm we consider.

\section{Preliminaries}
\label{prelimsection}
In each of the problems considered here, we categorized particles as being either in the {\em structure} built for the operation or as {\em free} particles.  
\begin{definition} 
At any point during the execution of the algorithm the {\em structure} refers to the set of particles recruited for use in some operations and assigned a specific role and position for that operation. They are in one of the states $\{${\em seed, matrix, vector, counter, prestop, result}$\}$.
\end{definition}
\begin{definition}
  At any point during the execution of the algorithm, the set of free particles consists of those particles that are not yet assigned a specific purpose. They are in one of the states $\{${\em leader, follower, inactive}$\}$. 
  \end{definition}
  Free particles may eventually become part of the structure or remain available for other uses. As free particles they actively move to make themselves available to extend the structure if needed, but may continue moving indefinitely if they are not recruited. Particle states are defined as the corresponding algorithms are presented in Sections~\ref{sec:binaryparticlecounter},~\ref{sec:matrixmult}, and~\ref{sec:edgedetect}.
  
  \emph{Tokens} are small structures (of constant size) of data which are held by exactly one particle at a time during their existence. Respecting the particles' memory constraints, each particle holds at most a constant number of tokens at any time. Configurations and schedules are defined for a set of particles and will be used to analyze the progress of the entire system toward the final goal. 
  
    \begin{definition}
    A {\em configuration} of the particle system at a point in time consists of the set of state variables $P_j$ for each particle $j$, including position, current state, and tokens held.
    \end{definition}

    We use $p_C(t)$ to describe the position of token $t$ at configuration $C$: If particle $j$ holds $t$ in configuration $C$, then $p_C(t)=j$ (ownership of $t$ is indicated in $P_j$).  Tokens travel through a predefined sequence of nodes, regardless of which particles occupy those nodes during the execution of the algorithm.
   \begin{definition}\label{tokenpathdef}
     A \emph{token path} of length $m$ is a set of particles $P_{k_1},P_{k_2},...,P_{k_m}$ such that $P_{k_l}$ is adjacent to $P_{k_{l+1}}$ and one or more tokens travel from $P_{k_x}$ to $P_{k_y}$ passing through only particles in the path for some $x,y$ with $1\leq x <y \leq m$.
     \end{definition}

    We consider a configuration $C$ to be valid if the system is connected (including both the structure and free particle set) and each particle is either contracted or expanded into adjacent positions with no single position occupied by two particles. When clear from context, we will refer to the particle $j$ and $P_j$ indistinctly.

In an asynchronous execution, the system progresses through a sequence of \emph{asynchronous rounds} (Definition~\ref{def:asynchround}). When a particle $P_j$ is activated during an asynchronous round, if it holds a token $t$ it can pass $t$ to any neighbor which has available token capacity at the time of the current activation of $P_j$.

\section{Binary Particle Counter Algorithm}\label{sec:binaryparticlecounter}
The first computational application of the amoebot model we analyze is a binary counter. The binary counter we describe here will also be used as a primitive for the matrix-vector multiplication algorithm presented in Section~\ref{sec:matrixmult}. In this implementation, the system contains only the seed particle and a set of initially inactive particles, already forming a line with the seed at the end at round $0$.\footnote{If a line of particles is not readily available, one can easily build one following the algorithm presented in~\cite{DBLP:conf/dna/DerakhshandehGS15} concurrently with the binary counting procedure -- i.e., there is no need for synchronization of the phases, as it happens in the matrix-vector multiplication algorithm presented below.} We denote the non-seed particles $P_0,...,P_{n-1}$ such that $P_0$ is a neighbor of the seed and labeling follows the line of particles moving away from the seed. Each non-seed particle represents a digit of the counter, with the particle in line closest to the seed representing the least significant bit of the counter. The value of the system as a whole can then be calculated using the state of each digit particle to determine the value it represents. In this analysis we consider both the number of tokens held by a particle $P_j$, denoted $P_j.count$, and a display property the particle computes based on state, denoted $P_j.display$. Each $P_j$ with $j<n-1$ receives tokens only from  $P_{j-1}$ (or $S$ if $j=0$). The method $P_j.tokenAvailable()$ returns true if there is a token in the shared memory for $P_j$ and $P_{j-1}$. If such a token exists, the method $P_j.takeToken()$ allows $P_j$ to pick it up. The method $P_j.canSendToken()$ determines if $P_j$'s shared memory with $P_{j+1}$ has space available. If so, $P_j.sendToken()$ allows $P_j$ to put a token in shared memory with $P_{j+1}$. Finally, $P_{j}.discardToken()$ removes the token from memory completely. As described in Algorithm~\ref{binarycounter}, when $P_j$ reaches its token capacity, here defined as two, it discards one token ($P_j.discardToken()$) and attempts to send the other, representing a carryover, to $P_{j+1}$, using $P_j.sendToken()$ (conditional on $P_j.canSendToken()$). 
  \begin{algorithm}[H]
\caption{Binary Counter Particle $P$}\label{binarycounter}
\begin{algorithmic}[1]
\Procedure{BinaryCounter}{}
\If{$P.count<2$ and $P.tokenAvailable()$}
\State $P.takeToken()$
\State $P.display = 1-P.display$
\State $P.count= P.count+1$
\EndIf
\If{$P.count = 2$ and $P.canSendToken() = true$}
\State $P.sendToken()$
\State $P.discardToken()$
\State $P.count = 0$
\EndIf
\EndProcedure
\end{algorithmic}
\end{algorithm}

The seed behaves as an interface to the counter, as described in Algorithm~\ref{binarycounterseed}.  It receives activations from an external source to increment the counter, upon which it constructs new tokens and sends those to $P_0$ if  there is space in the shared memory with $P_0$, i.e. when $S.canSendToken()$ returns true.
 \begin{algorithm}[H]
\caption{Binary Counter Seed Particle $S$}\label{binarycounterseed}
\begin{algorithmic}[1]
\Procedure{BinaryCounterSeedActivated}{}
\If{$P.canSendToken() = true$}
\State Create new token $t$
\State $S.sendToken()$
\EndIf
\EndProcedure
\end{algorithmic}
\end{algorithm}


\subsection{Parallel Schedules}
     
All of our algorithms, presented in Sections~\ref{sec:binaryparticlecounter},~\ref{sec:matrixmult}, and~\ref{sec:edgedetect}, follow an asynchronous execution. However, for the analyses of these algorithms, we considered  executions according to parallel schedules, since those are easier to handle and will provide a worst-case scenario in terms of number of rounds for asynchronous schedules. In a parallel execution, the system progresses through a sequence of \emph{parallel rounds}. 

\begin{definition} \label{parallelround}
During one \emph{parallel round} starting with configuration $C$ and resulting in configuration $C^*$, one of the following is true for each particle $p$:
\begin{enumerate}
\item $p$ occupies the same node(s) in $C$ and $C^*$,
\item $p$ occupies one node in $C$ and expands to an additional adjacent node during the round,
\item $p$ occupies two adjacent nodes in $C$ and contracts to a single node during the round, leaving the other node empty in $C^*$, or
\item $p$ occupies two adjacent nodes in $C$ and contracts in a handover such that in $C^*$ a different particle has expanded into one of the nodes $p$ occupied in $C$.
\end{enumerate}
Additionally, for each token $t$, let $P_k$ be such that $k = p_C(t)$. Then at the end of the parallel round one of the following is true:\begin{enumerate}
\item $p_{C^*}(t) = p_{C}(t)$,
\item if a particle $P_{k'}$ adjacent to $P_k$ is below capacity in $C$, $p_{C^*}(t) =k'$, or 
\item  if there is a token path length $d$ (labeled as particles $P_{k_1},...,P_{k_d}$),  for each $1\leq l \leq d-1$ the particle  $P_{k_l}$ in the path has a token $t_{l}$ (such that $t =t_l$ for some $l$) which needs to move to $P_{k_{l+1}}$, and $P_{k_l}$ has available token capacity, then $p_{C^*}(t_l) = p_{C}(t_l)+1 $ for each $1\leq l \leq d-1$.
\end{enumerate}
 \end{definition}
\begin{definition} \label{syncdef}
A movement schedule $(C_0,C_1,...C_f)$ is a \emph{parallel schedule} if each $C_i$ is a valid configuration and for each $i \geq 0$, $C_{i+1}$ is reached from $C_i$ in exactly one parallel round. 
 \end{definition}
  In asynchronous execution, the system progresses through a sequence of \emph{particle activations}, meaning only one particle is active at a time. When activated, a particle can perform an arbitrary bounded amount of computation (including passing tokens) and make at most one movement. An \emph{asynchronous round} is the elapsed time until each particle has been activated at least once. When a particle $P$ is activated, if it holds a token $t$ it can pass $t$ to any neighbor which has available token capacity at the time of the current activation of $P$.
 
 \begin{definition}
 \label{asyncdef}
 A movement schedule $(C_0,C_1,...C_f)$ is an \emph{asynchronous schedule} if each $C_i$ is a valid configuration and for each $i \geq 0$, $C_{i+1}$ is reached from $C_i$ by execution of one asynchronous round.  
  \end{definition}

\subsection{Runtime Analysis}\label{sec:analysisbinaryparticle}
In this section we provide a brief, high-level sketch of the proof that shows that a counter with $n$ particles can count to $v$ (where $v\leq 2^n-1$) in $\Theta(v)$ asynchronous rounds (the proofs and more details can be found in the Appendix). Our proof relies on a dominance argument of parallel schedules by asynchronous executions, in terms of the number of rounds needed to complete the algorithm. 

\begin{lemma}\label{asyncdominates}
For any asynchronous particle activation sequence $A$, there exists a parallel schedule ${\cal P}$ such that the number of asynchronous rounds needed by the binary counter algorithm according to $A$ is at most equal to the number of parallel rounds required by the algorithm following ${\cal P}$.
\end{lemma}

We can then count the total number of bit flips that occur in the counter to get the result:
\begin{lemma}\label{counter:syncOn}
The parallel binary counter algorithm counts to the value $v$ in $O(v)$ parallel rounds.
\end{lemma}
Combining these two results, we get:
\begin{theorem}\label{counter:asyncOn}
The asynchronous binary counter counts to the value $v$ in $\Theta(v)$ asynchronous rounds.
\end{theorem}

\section{Particle Matrix Multiplication Algorithm}\label{sec:matrixmult}
	The next computational problem we solve using the amoebot model is matrix-vector multiplication. As before, the seed acts as a source of external input into the system. We suppose the system is initially unaware of the dimensions or values of the matrix and vector to be multiplied, so they will enter the system through the seed particle. The stream of information entering the system from the seed can contain values of matrix or vector entries (we assume each fits on a single particle), end of column markers, and end of vector markers. The seed particle at no point computes the dimensions of the problem since it receives values online in sequence from an external source. The seed then passes values, encapsulated in tokens, into the system as the algorithm proceeds. 
			
	Let $A$ be a $h\times w$ matrix and $\vec{x}$ be a $w\times 1$ vector for some nonzero integers $h$ and $w$. The result of the matrix vector multiplication $A\vec{x}$ is then $\vec{b}$, which is stored using a set of counters described in Section~\ref{setupsection}. The problem is streamed into the system in the order: values for each matrix column from top to bottom, left to right, followed by the values of $\vec{x}$ ordered from top to bottom. As vector values reach their final positions, vector particles also generate result counter tokens, which are passed along to determine how many particles should position themselves to store the results of the multiplication. As shown in Figure~\ref{matrixsetupfig}a., particles assigned to represent values of $\vec{x}$ are positioned across the top of those representing matrix $A$, such that the line of matrix particles directly below a vector particle is the corresponding matrix column. The vector value is then passed down the column and used by each matrix particle it reaches to produce an individual product. Products are then passed across the row of matrix columns to where the set of result particles are positioned to store the product totals.
		
	This algorithm can also be extended to complete matrix-matrix multiplication. To multiply matrices $A$ and $C$, the setup is the same as before but with the first column of $C$, $\vec{c_0}$ replacing the vector $\vec{x}$. If $C$ has a width of $y$, after each column $\vec{c_i}$ is multiplied by $A$, for $i<y$, we add a new set of results particles to store the vector $\vec{b}_i$. Thus the entire result matrix $B$ can be stored as series of vectors $\vec{b}_0,\vec{b}_1,...,\vec{b}_{y-1}$, as shown in Figure~\ref{matrixsetupfig}b.

	\begin{figure}[H]		\centering 
	\includegraphics[width=0.5\textwidth]{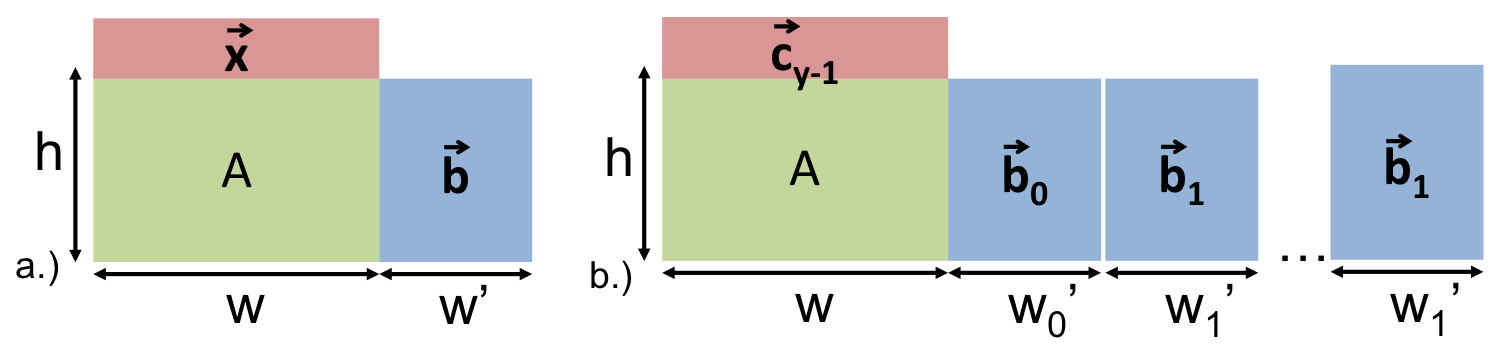} 
					\caption{a.) General matrix-vector multiplication A$\vec{x}$ setup for $h\times w$ matrix $A$ and $w\times 1$ vector $\vec{x}$. b.) General matrix-matrix multiplication $AC$ for $h\times w$ matrix $A$ and $w\times y$ matrix $C$. Shown during final matrix-vector multiplication $A\vec{c}_{y-1}$.}\label{matrixsetupfig}

\end{figure}
The matrix, vector, and result particles do not know their indices relative to the whole system but can orient themselves such that they know which direction is across the matrix row and which direction is down the matrix column. To multiply a matrix by multiple vectors in a stream, this setup only needs to be executed once. If a finished notification is sent to the seed after each matrix-vector multiplication completes, an additional vector can be used without any changes to the matrix.

\subsection{Description and Correctness}\label{setupsection}
	We refine the notation of a configuration from Section~\ref{prelimsection} to specify the particles' functions in the final system. Let \\
	$C_i = (M_{0,0},M_{0,1}...M_{0,w-1},M_{1,0}...M_{h-1,w-1},$\\$R_{0,0},R_{0,1},...,R_{0,w'},...R_{1,0}...R_{h,w'},V_{0},V_1...V_{w-1})$\\
	 be the configuration at round $i$ where $M_{u,v}$ is the configuration of the particle which will eventually be the matrix particle at position $(u,v)$, $R_{u,q}$ will be a result particle at position $(u,q)$ in the results matrix, and $V_v$ is the vector particle at index $v$ in vector $\vec{x}$. Let $c$ be the token capacity of matrix, vector and result particles, and let $m$ be the maximum value of a matrix or vector entry. We then use $w'$ to denote the number of columns of results particles constructed, so $0\leq u<h$, $0\leq v <w$, and $0\leq q < w'$. Enough result columns are constructed to hold the maximum possible number of tokens generated, so $w' = \lceil{\log_c(m^2w)}\rceil$. Finally, we denote the minimum number of particles necessary to complete setup as $n'$, so $n' = hw+w+hw'$. Since particles are given tasks on a first-come, first-serve basis, particles that remain free particles throughout execution do not have any effect on the correctness of the system. Particles are categorized in configurations based on their final location, but are all initially free particles and begin by executing the spanning forest algorithm in~\cite{DBLP:conf/dna/DerakhshandehGS15}, making their initial states leaders and followers (for completeness, we present the spanning forest algorithm in Appendix~\ref{spanningforestappendix}).
	
Tokens travel in a predetermined direction in the set of matrix, vector, and result particles. For clarity, we extend the range of the position function $p(t)$ for token $t$ to be ordered pairs representing position in a two-dimensional arrangement of system particles. 

Particles are recruited to the system from a set of unassigned particles, the free particles, whose motion is determined by a spanning forest algorithm. The only particle not initially a free particle is the seed which is described in Appendix~\ref{appendix:matrixcode}. Recruited particles then hold the streamed values and become fixed in a position for the matrix and vector, and later for the results. This is done by setting flags from the seed, vector, and matrix particles which point to where a new particle is needed, so a free particle will become part of the structure when one of these flags points to it. Result particles are similarly recruited by setting flags to point to where a particle may be needed based on the maximum possible values of the matrix and vector, but result particles have the option to leave the structure after multiplication has completed if they are not needed to represent the result.

\begin{figure}
	\includegraphics[width=0.3\textwidth]{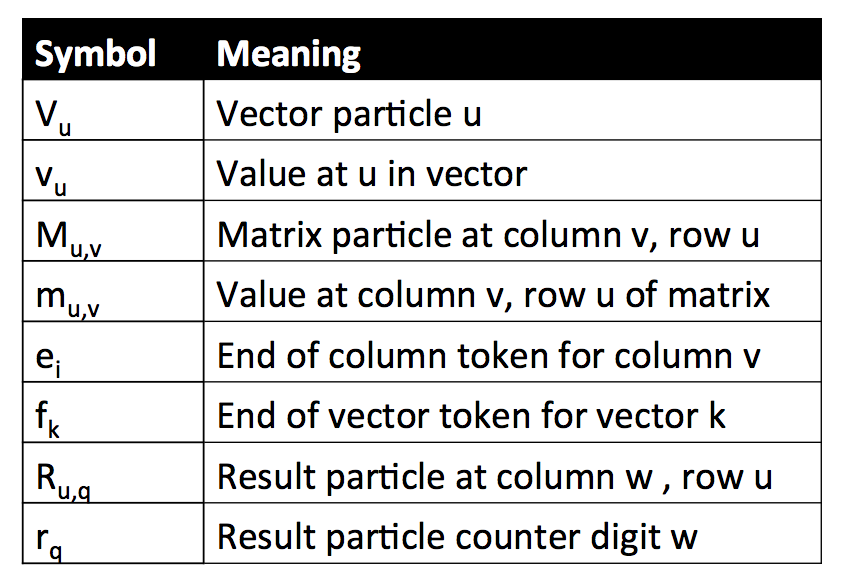}
\\	a.) \includegraphics[width=0.4\textwidth]{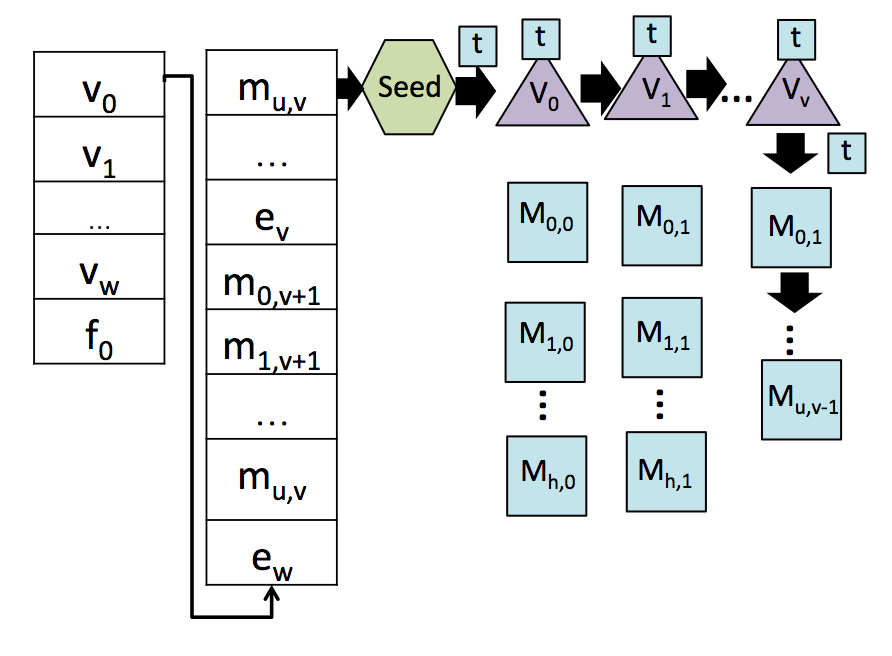}
	\\b.) \includegraphics[width=0.3\textwidth]{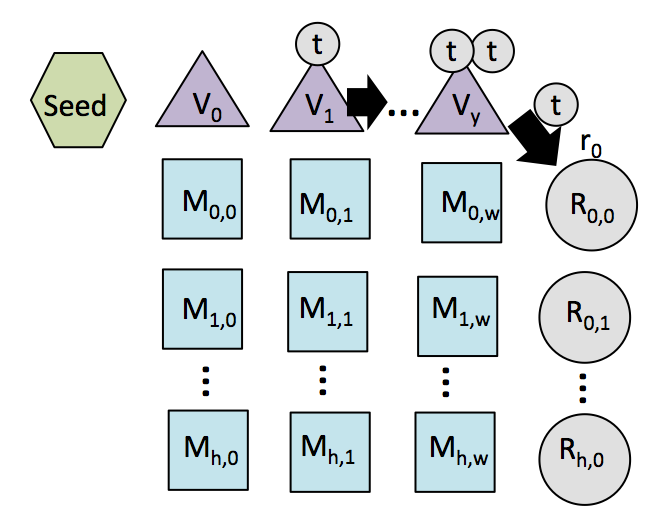}
			\caption{Diagram of system setup and notation used. Shapes labeled $V_v,M_{u,v}$, or $R_{u,q}$ are particles and small squares/circles labeled $t$ are tokens, colored and shaped based on their final destination. Free particles are not shown.}
		\label{setupfig}
\end{figure}

	Figure~\ref{setupfig}a. conceptually shows a system in the process of executing the setup algorithm. At the depicted point in time, each of the matrix values $m_{0,0},m_{1,0}...m_{h,0},m_{0,1},....m_{u-1,v}$ has been streamed into the system through the seed, and assigned to a corresponding particle by passing matrix tokens (squares labeled $t$) into the structure. For example, the value $m_{0,0}$ is assigned to particle $M_{0,0}$ at the upper left corner of the matrix. Note that any notions of ``up/down'' and ``left/right'' are relative to the orientation passed to the system from the seed particle, and do not assume any absolute orientation of the system. The next value to be added to the matrix, $m_{u,v}$ is shown at the head of the stream of values entering the seed particle. It will then be passed across the row of vector particles ($V_0,...V_v$) to the furthest particle, $V_v$, that has been recruited so far. That vector particle, $V_v$, will then pass $m_{u,v}$ down the column $M_{0,v},...M_{u-1,v}$. The particle $M_{u-1,v}$ is then responsible for recruiting a new matrix particle from the set of free particles (not shown) to be $M_{u,v}$ and hold the value $m_{u,v}$. This process will continue until the last column is completed.

\begin{figure}
	\centerline{\includegraphics[width=0.5\textwidth]{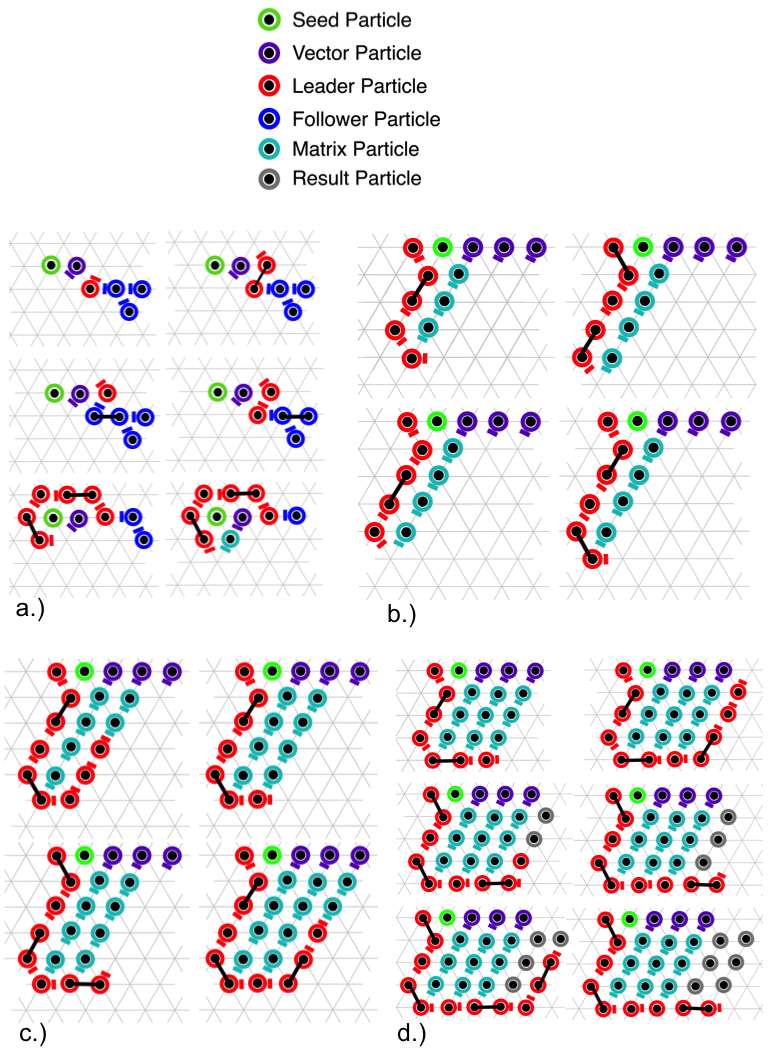}}
			\caption{a.) Setup of first matrix particle $M_{0,0}$ in worst case. b.) Setup of second matrix particle column. c.) Setup of third matrix particle column, generalizes to all subsequent columns. d.) Setup of result particle columns. (Lemma~\ref{linearcolfill})}
		\label{setupcol}
\end{figure}

The last part of the value stream, shown in the left half of the stream entering the seed in Figure~\ref{setupfig}a., is the set of vector values. Vector values are assigned to the first vector particle they reach which does not yet have a value. As each vector value is assigned, a result counter token is generated and passed down the vector away from the seed.  In Figure~\ref{setupfig}b. these are the circular tokens which are passed from $V_{w-1}$ to $R_{0,0}$ such that $R_{0,0},R_{1,0},...$ acts as a counter. When the farthest vector particle receives or generates result counter tokens it begins to recruit particles to start forming the result segment of the structure.  When using multiple matrix-vector multiplications to perform a matrix-matrix multiplication, the existing result particles at the end of each matrix-vector product stop performing operations other than passing tokens. Then new sets of result particles are recruited for each matrix-vector multiplication in the sequence. Note that all phases of the algorithm are running concurrently, and there is no synchronization between phases. In order to prove the correctness and runtime of our algorithm, we will show that the different phases of our algorithm eventually correctly terminate in order.

Once the first end of vector marker, $f_0$ is received by the seed, setup will be completed. Figure~\ref{setupcol} shows the setup process, including the initial movement of free particles, on the actual amoebot model. When setup is finished, the multiplication can be executed. The matrix-vector multiplication can be summarized by the following steps:
\begin{enumerate}
\item each vector particle $V_u$ passes its value $v_u$ in a token to matrix particle $M_{u,0}$, i.e. directly downward,
\item each matrix particle $M_{u,v}$ with value $m_{u,v}$ computes the product $m_{u,v}\cdot v_{u}$,
\item $M_{u,v}$ passes the vector value $v_u$ to $M_{u,v+1}$ (if $M_{u,v+1}$ exists) so the vector value continues to move down the column,
\item $M_{u,v}$ passes a total of $m_{u,v}\cdot v_{u}$ result counter tokens to $M_{u,v+1}$ (or $R_{u,0}$ if $M_{u,v+1}$ does not exist), i.e. to the right across the row, and
\item each result particle $R_{u,v}$ accepts result counter tokens until at capacity, and then clears its counter and passes a carry over token to $R_{u+1,v}$ (executing Algorithm~\ref{binarycounter} relative to its row of result particles).
\end{enumerate}
 Once multiplication has completed, the excess particles recruited to be result particles can be released back to being free, so that the final system configuration is minimal. Detailed pseudocode descriptions of the algorithms can be found in Appendix~\ref{appendix:matrixcode} and a proof of the following theorem appears in Appendix~\ref{appendix:matrixmult}. Snapshots of an implementation run of our algorithm appear in Figure~\ref{setupcol}.


\begin{theorem}
\label{streamingcorrect}
The streaming setup model successfully completes with the values in the correct locations as long as the system contains enough particles.
\end{theorem}

\subsection{Runtime Analysis}\label{sec:analysismatrixmult}

Similarly to the binary counter case, to show bounds on the runtime of the matrix multiplication system, we show bounds for a parallel schedule (defined in Section~\ref{syncdef}) and that such a parallel schedule is dominated by the asynchronous schedule. For comparisons of progress in a system, we look at how close particles and tokens are to their final position nodes of the graph. We give a high-level sketch of the proof here; please see Appendix~\ref{appendix:matrixmult} for the full proof.

For matrix value tokens, final position is the particle in the factor matrix corresponding to the value. For vector value tokens final position is the bottom matrix particle in the column under the vector particle corresponding to their value. For product tokens final position is in the counter representing the value of the result vector corresponding to the matrix row in which the product token originated.
For the matrix multiplication problem, define the operator $\preceq_M$ between configurations to roughly mean that $C \preceq_M C'$ if and only if tokens and particles are further along their paths towards their final destination in C' than in C. Using this comparison operator, we show:
\begin{lemma}\label{setupdomination}
For any asynchronous particle activation sequence $A$, there exists a parallel schedule ${\cal P}$ such that the number of asynchronous rounds needed by the matrix-vector multiplication algorithm according to $A$ is at most equal to the number of parallel rounds required by the algorithm following ${\cal P}$.
\end{lemma}
In this section we first consider the setup phase of the matrix-vector multiplication, in which particles move from initial spanning tree configurations to the structure configuration consisting of $matrix$, $vector$, and $result$ particles as previously described. Setup also includes the passing of tokens corresponding to matrix and vector values, but not the passing of tokens corresponding to products of these values. To show that system setup completes in $O(n')$ parallel rounds, we first show that our modified spanning tree primitive supplies particles to construction as necessary, so that we have:
\begin{lemma}
\label{columnlinearlemma}
Each matrix and result particle column takes $O(h)$ rounds to fill with particles in the parallel execution.
\label{linearcolfill}
\end{lemma}
Then with $w+w' = O(w)$ columns which need to be filled, we can show:
\begin{theorem}
\label{matrixOnp}
The parallel matrix system setup completes in $O(n')$ rounds.
\end{theorem}
\begin{theorem}
\label{matrixthetanp}
The streaming matrix system setup completes in $\Theta(n')$ rounds.
\end{theorem}

We next consider the actual matrix-vector multiplication process. Multiplication is initiated by each vector particle sending a token representing the value corresponding to its position down the column of the matrix, such that it is seen by the matrix particle at each position which directly multiplies with that vector value. The amount of computation for the multiplication step is bounded by the time for tokens to travel down matrix columns and across matrix and result rows, so we have:
 \begin{lemma}\label{matrix:syncOn}
The parallel matrix-vector multiplier completes in $O(h+w)$ rounds.
\end{lemma}
\begin{theorem}
\label{sumthetahw}
The asynchronous matrix-vector multiplier completes calculations in $\Theta(h+w)$ rounds.
\end{theorem}
\begin{theorem}
\label{sumthetahwy}
The asynchronous matrix-matrix multiplier completes calculations in $\Theta(y(h+w))$ rounds.
\end{theorem}

\section{Image Processing Applications}\label{sec:edgedetect}
Both the setup and multiplication steps of the matrix-vector multiplication algorithm can be used in image processing applications. Individual particles can be assigned to store individual pixels or small grids of pixels of an image, and their proximity to particles holding the corresponding adjacent pixels makes a number of localized image processing algorithms highly efficient. In this section we first discuss using the amoebot model to execute the Canny edge detection algorithm on a single channel image, meaning with a single scalar value for each pixel. Pixel values are streamed into the system and a grid is established in the same way as in matrix-vector multiplication setup (Section~\ref{setupsection}), but without the requirement of result particles. Image pixel values are streamed as matrix values, so that the matrix particles store the image and independently start to execute the algorithm as soon as they receive a value. We next discuss how to use the amoebot model to execute image color transformations that use matrix-matrix multiplications. In this application, the image is represented as a matrix with entries corresponding to each color component of each pixel. A series of matrix-vector multiplications is then executed to complete color transformation.
\subsection{Canny Edge Detection}
Canny edge detection uses the Sobel operator, after initially applying Gaussian filtering, and then performs non-maximum suppression and hysteresis thresholding. The Sobel operator applies the matrix convolution operation to two $3\times 3$ kernels with the original image to approximate horizontal and vertical derivatives. The resulting values at each point are then combined and used to calculate the direction of the gradient, as described in Algorithm~\ref{pixelcanny}. This means each particle $P_j$ compares its own gradient magnitude $P_j.G$ with those of its two neighbors in the positive and negative directions determined by the angle of the gradient, $P_j.\theta$. If $P_j$ is not the local maximum on the gradient it not considered to potentially be part of an edge. Thresholds are then used to select particles with gradient values in a particular range. These particles, plus any local maxima particles in the neighborhood of those in the range, are considered to be the edge set. 
 \begin{algorithm}
\caption{Pixel Particle P}\label{pixelcanny}
\begin{algorithmic}[H]
\State $P.pixelValue$ = pixel value as read from the image
\State Label grid values around particle $P$ as $\begin{pmatrix} g_2 & g_1 & g_0\\g_3 & P & g_7 \\ g_4 & g_5 & g_6\end{pmatrix}$
\State $P.gridvalues[8] = \{g_0,g_1,...g_7\}$ are pixel values of the $8$ surrounding particles on a grid
\State $P.G$= gradient magnitude
\State $P.\theta$ = gradient direction
\State $P.edgeCertainty \in \{Sure,SureNeighbor,No\}$
\Procedure{Sobel}{}
\State Gather $P.gridvalues$ directly or by querying neighbors
\State  $G_x = K_x *P.gridvalues$
\State  $G_y = K_y*P.gridvalues$
\State $P.G = \sqrt{G_x^2+G_y^2}$
\State $\displaystyle P.\theta = \text{atan}( \frac{G_x}{G_y})$
\EndProcedure
\Procedure{CannyEdgeDetect}{}
\State Gather $P.gridvalues$ directly or by querying neighbors
\State  $P.pixelValue = K_G * P.gridvalues$
\State $P.noiseFinished =true$
\State Update $P.gridvalues$ after each has $P.noiseFinished = true$
\State Execute procedure \textsc{Sobel}
\State Switch on angle, compare to appropriate neighbors to determine if maximum
\If{$minThreshold<P.G<maxThreshold$}
\State $P.certainty = Sure$
\ElsIf{$P$ has neighbor $N$ with $N.certainty = Sure$}
\State $P.certainty = SureNeighbor$
\Else 
\State $P.certainty = No$
\EndIf
\EndProcedure
\end{algorithmic}
\end{algorithm}

We setup the image as a matrix using the algorithms defined in Section~\ref{setupsection}. Each step of edge detection requires particles to have information from other particles but only within a constant distance, so we have:
\begin{theorem}
\label{edgeconst}
Edge detection will complete in constant time after image setup.
\end{theorem}
\subsection{Color Transformations}
For color transformations, each pixel in the input image must be described in terms of the three color values. In this application, the input matrix has a row corresponding to each pixel of the original image and three columns corresponding to red, green, and blue. The transformation matrix is then streamed into the system as a sequence of vectors, each of which is multiplied by the matrix. The values in the transformation matrix determine the operation, such as filtering or saturation changes.
\section{Simulation Results}

As expected, Figure~\ref{allresults}a. shows that the number of rounds required for the binary counter to reach a value $v$ increases linearly with $v$. The results shown are each for a set of $10$ particles arranged in a line before the system begins to execute. Value counted is the number of distinct counter tokens fed into the system by the seed particle.

For matrix-vector multiplication, the experiments in Figure~\ref{allresults}b. show an approximately linear increase in the number of rounds for system setup and execution as the number of particles for the matrix-vector structure, $n'$, increases. 

In Figure~\ref{allresults}c. we show two examples of edge detection on small images. The implementation uses only a $3\times 3$ Gaussian kernel created with a value of $\sigma = 1$ for blurring. It also discards an outer boundary at each step rather than using an inference method to fill in nonexistent values around edge pixels, so the images are padded with borders of zero-value pixels before inputted into the system. Figure~\ref{allresults}c. shows the results of edge detection in a simple $10\times 10$ shape and a more complex $16\times 16$ image of a coin. Only the red values from the RGB information in the coin image were used.

Figure~\ref{allresults}d. shows the results of color transformations by multiplying an image matrix by a $3\times 3$ operator. The upper right example shows increased saturation, the bottom left shows conversion to grayscale, and the lower right shows color filtering.
\begin{figure}[h]
	\centering
	\includegraphics[width=0.5\textwidth]{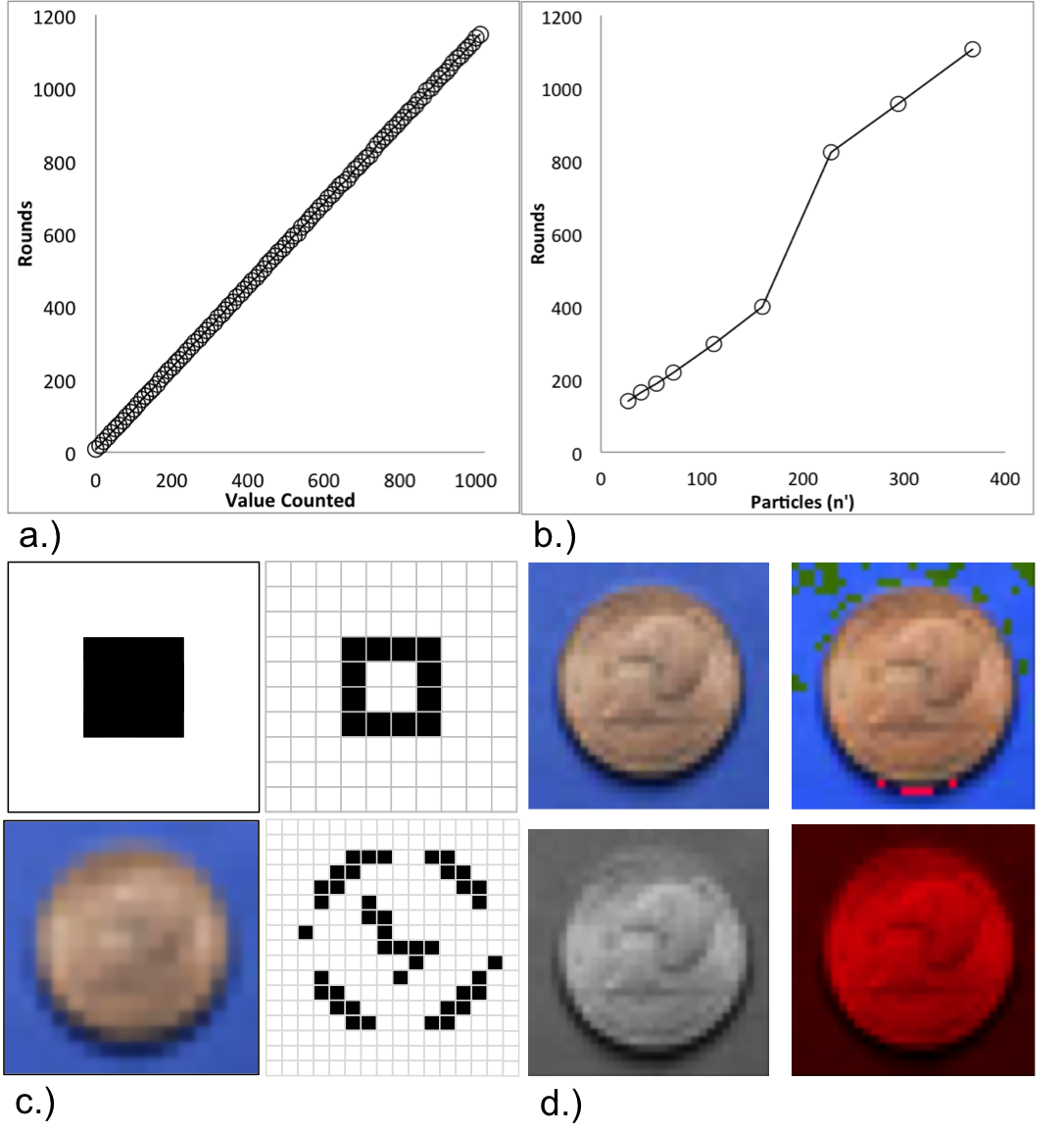}
		\caption{a.) Number of asynchronous rounds per value of $v$ counted in the binary counter; b.) number of asynchronous rounds per vector dimension in matrix-vector multiplication; c.) edge detection results; d.) color transformation results}
		\label{allresults}
\end{figure}

\section{Discussion}

We have described basic computational algorithms that can be used in much larger computing applications. Due to the limitations of the system receiving input through a seed particle, the binary counter requires $\Theta(v)$ asynchronous rounds to count to a value of $v$. The setup of the matrix-vector multiplication system is similarly limited by the input and time to assemble the structure of particles, so it requires $\Theta(n')$ rounds to setup the $n'$ particles used to represent the matrix, vector, and the vector of the product. However, the actual matrix-vector multiplication operations benefit from the parallelism of the system and each matrix-vector multiplication requires only $\Theta(h+w)$ asynchronous rounds (recall that $h$ is the matrix height and $w$ is the matrix width). This is especially beneficial for a matrix-matrix multiplication which requires only one execution of the setup algorithm (excluding the setup of additional results particles) to multiply an input matrix by each column of the other input matrix.

Possible improvements on the setup time limitation to the efficiency could include using more particles than will actually be needed to more quickly approximate the structure for the system. For instance, it may be useful to apply compression, as defined by~\cite{cannon2016markov} to be minimizing the outermost perimeter of the configuration of particles. If a compression algorithm produced a particle configuration that contained the necessary structure for one of our algorithms, it may be more efficient to execute compression and then select the useful structure as a subset of the compressed system.

There are many other applications that could use our basic counting and matrix multiplication algorithms and setup in a programmable matter context. 
\bibliographystyle{elsarticle-num}
{\small \bibliography{citations}}

\begin{thebibliography}{10}
\expandafter\ifx\csname url\endcsname\relax
  \def\url#1{\texttt{#1}}\fi
\expandafter\ifx\csname urlprefix\endcsname\relax\def\urlprefix{URL }\fi
\expandafter\ifx\csname href\endcsname\relax
  \def\href#1#2{#2} \def\path#1{#1}\fi

\bibitem{Toffoli1991263}
T.~Toffoli, N.~Margolus, Programmable matter: Concepts and realization, Physica
  D: Nonlinear Phenomena 47 (1991) 263 -- 272.

\bibitem{DBLP:conf/spaa/DerakhshandehGR16}
Z.~Derakhshandeh, R.~Gmyr, A.~W. Richa, C.~Scheideler, T.~Strothmann, Universal
  shape formation for programmable matter, in: {ACM} {SPAA}, 2016, pp.
  289--299.

\bibitem{DBLP:conf/dna/DerakhshandehGP16}
Z.~Derakhshandeh, R.~Gmyr, A.~Porter, A.~W. Richa, C.~Scheideler,
  T.~Strothmann, On the runtime of universal coating for programmable matter,
  in: DNA 22., 2016, pp. 148--164.

\bibitem{DBLP:conf/dna/DerakhshandehGS15}
Z.~Derakhshandeh, R.~Gmyr, T.~Strothmann, R.~A. Bazzi, A.~W. Richa,
  C.~Scheideler, Leader election and shape formation with self-organizing
  programmable matter, in: DNA 21., 2015, pp. 117--132.

\bibitem{DBLP:conf/podc/CannonDRR16}
S.~Cannon, J.~J. Daymude, D.~Randall, A.~W. Richa, A markov chain algorithm for
  compression in self-organizing particle systems, in: {ACM} {PODC}, 2016, pp.
  279--288.

\bibitem{lynch1996distributed}
N.~A. Lynch, Distributed algorithms, Morgan Kaufmann, 1996.

\bibitem{adleman1994molecular}
L.~M. Adleman, Molecular computation of solutions to combinatorial problems,
  Nature 369 (1994) 40.

\bibitem{doty2012theory}
D.~Doty, Theory of algorithmic self-assembly, Communications of the ACM 55~(12)
  (2012) 78--88.

\bibitem{patitz2014introduction}
M.~J. Patitz, An introduction to tile-based self-assembly and a survey of
  recent results, Natural Computing 13~(2) (2014) 195--224.

\bibitem{woods2015intrinsic}
D.~Woods, Intrinsic universality and the computational power of self-assembly,
  Phil. Trans. R. Soc. A 373~(2046) (2015) 20140214.

\bibitem{rothemund2000program}
P.~W. Rothemund, E.~Winfree, The program-size complexity of self-assembled
  squares, in: {ACM}{STOC}, ACM, 2000, pp. 459--468.

\bibitem{arbuckle2010self}
D.~Arbuckle, A.~A. Requicha, Self-assembly and self-repair of arbitrary shapes
  by a swarm of reactive robots: algorithms and simulations, Autonomous Robots
  28~(2) (2010) 197--211.

\bibitem{das2010computational}
S.~Das, P.~Flocchini, N.~Santoro, M.~Yamashita, On the computational power of
  oblivious robots: forming a series of geometric patterns, in:
  {SIGACT-SIGOPS}, ACM, 2010, pp. 267--276.

\bibitem{kernbach2013handbook}
S.~Kernbach, Handbook of collective robotics: fundamentals and challenges, CRC
  Press, 2013.

\bibitem{woods2013active}
D.~Woods, H.-L. Chen, S.~Goodfriend, N.~Dabby, E.~Winfree, P.~Yin, Active
  self-assembly of algorithmic shapes and patterns in polylogarithmic time, in:
  {ITCS}, ACM, 2013, pp. 353--354.

\bibitem{DBLP:conf/nanocom/DerakhshandehGR15}
Z.~Derakhshandeh, R.~Gmyr, A.~W. Richa, C.~Scheideler, T.~Strothmann, An
  algorithmic framework for shape formation problems in self-organizing
  particle systems, in: NANOCOM, 2015, pp. 21:1--21:2.

\bibitem{fox1987matrix}
G.~C. Fox, S.~W. Otto, A.~J. Hey, Matrix algorithms on a hypercube {I}: Matrix
  multiplication, Parallel computing 4~(1) (1987) 17--31.

\bibitem{cannon1969cellular}
L.~E. Cannon, A cellular computer to implement the kalman filter algorithm.,
  Tech. rep., DTIC Document (1969).

\bibitem{van1997summa}
R.~A. Van De~Geijn, J.~Watts, \uppercase{SUMMA}: Scalable universal matrix
  multiplication algorithm, Concurrency-Practice and Experience 9~(4) (1997)
  255--274.

\bibitem{choi1997new}
J.~Choi, A new parallel matrix multiplication algorithm on distributed-memory
  concurrent computers, in: HPC Asia'97, IEEE, 1997, pp. 224--229.

\bibitem{sobel1990isotropic}
I.~Sobel, An isotropic 3$\times$ 3 image gradient operator, Machine vision for
  three-dimensional scenes (1990) 376--379.

\bibitem{geusebroek2001color}
J.-M. Geusebroek, R.~Van~den Boomgaard, A.~W.~M. Smeulders, H.~Geerts, Color
  invariance, IEEE Transactions on Pattern Analysis and Machine Intelligence
  23~(12) (2001) 1338--1350.

\bibitem{haeberli1993matrix}
P.~Haeberli, Matrix operations for image processing, URL: http://www. sgi.
  com/grafica/matrix/index. html.

\bibitem{canny1986computational}
J.~Canny, A computational approach to edge detection, IEEE Transactions on
  pattern analysis and machine intelligence~(6) (1986) 679--698.

\bibitem{ogawa2010efficient}
K.~Ogawa, Y.~Ito, K.~Nakano, Efficient canny edge detection using a gpu, in:
  ICNC, IEEE, 2010, pp. 279--280.

\bibitem{gentsos2010real}
C.~Gentsos, C.-L. Sotiropoulou, S.~Nikolaidis, N.~Vassiliadis, Real-time canny
  edge detection parallel implementation for \uppercase{FPGA}s, in: ICECS,
  IEEE, 2010, pp. 499--502.

\bibitem{neoh2004adaptive}
H.~S. Neoh, A.~Hazanchuk, Adaptive edge detection for real-time video
  processing using \uppercase{FPGA}s, Global Signal Processing 7~(3) (2004)
  2--3.

\bibitem{zhang2011fsim}
L.~Zhang, L.~Zhang, X.~Mou, D.~Zhang, Fsim: A feature similarity index for
  image quality assessment, IEEE transactions on Image Processing 20~(8) (2011)
  2378--2386.

\bibitem{bae2016novel}
S.-H. Bae, M.~Kim, A novel image quality assessment with globally and locally
  consilient visual quality perception, IEEE Transactions on Image Processing
  25~(5) (2016) 2392--2406.

\bibitem{plataniotis2013color}
K.~Plataniotis, A.~N. Venetsanopoulos, Color image processing and applications,
  Springer Science \& Business Media, 2013.

\bibitem{cannon2016markov}
S.~Cannon, J.~J. Daymude, D.~Randall, A.~W. Richa, A markov chain algorithm for
  compression in self-organizing particle systems, in: Proceedings of the 2016
  ACM Symposium on Principles of Distributed Computing, ACM, 2016, pp.
  279--288.

\end{thebibliography}

\clearpage
\appendix

\subsection{Spanning Forest Primitive}\label{spanningforestappendix}
The \emph{spanning forest primitive} is used at the beginning of system execution to organize particles into a spanning forest \textit{F}, which defines how particles can move while preserving connectivity. At round $0$, particles are idle, except the seed particle. 
\begin{algorithm}
\caption{Spanning Forest Primitive}
\label{alg:spanningForestAlgorithm}
\begin{algorithmic}
\State Define the {\em structure} as the set of particles whose state is in $\{$ {\em seed,matrix,vector,prestop,result}$\}$ (Section~\ref{prelimsection})
\State  A particle $P$ acts depending on its state:
\State  \textbf{idle}:  If $P$ is adjacent to the structure, it becomes a \emph{leader} particle. If a neighbor $P'$ is a leader or a follower, $P$ sets the flag $P.parent$ to the label of the port to $P'$ and becomes a \emph{follower}.
        If none of the above applies, $P$ remains idle.
          
    \State    \textbf{follower}: If $P$ is contracted and adjacent to the structure then $P$ becomes a \emph{leader} particle. If $P$ is contracted and has an expanded parent, then $P$ initiates a $handover$ motion; otherwise, if $P$ is expanded, it considers the following two cases: $(i)$ if $P$ has a contracted child particle $q$, then $P$ initiates $handover$ motion; $(ii)$ if $P$ has no children and no idle neighbor, then $P$ contracts.
\State
        \textbf{leader}: 
        Particle $P$ will expand when possible to move around the perimeter of the structure in the counterclockwise direction, until reaching a position to become a member of the structure (Algorithm~\ref{freeparticle}).      \\ 

      \end{algorithmic}
\end{algorithm}

	\subsection{Matrix-Vector Multiplication Pseudocode}
\label{appendix:matrixcode}

 \begin{algorithm}
\caption{Free Particle P}\label{freeparticle}
\begin{algorithmic}[H]
\State $P.neighbors = $  neighbor set
\State $P.state \in \{Leader,Prestop,Vector,Matrix,Result\}$, state of the particle
\Procedure{JoinStructure}{}
\If{$P.position$ will eventually be a structure position  }
\State $p.state = Prestop$
\EndIf
\If{ $N\in P.neighbor$ with $N.vectorFlag$ pointing to $P$ }
 \State $P.state = Vector$
\EndIf
\If{ $N\in P.neighbor$ with $N.matrixFlag$ pointing to $P$ }
\State $P.state = Matrix$
\EndIf
\If{ $N\in P.neighbor$ with $N.resultFlag$ pointing to $P$ }
\State $P.state = Result$
\EndIf
\EndProcedure
\end{algorithmic}
\end{algorithm}

 \begin{algorithm}
\caption{Seed Particle S}\label{seedparticlesetup}
\begin{algorithmic}[H]
\State $S.mode \in \{matrix,vector,none\}$, current type of value being passed from stream
\State $S.right =$ neighbor to right (across row)
\Procedure{SeedParticleSetup}{}
\State $S.mode = matrix$
\Switch{Type of value $v$ received}
	\Case{$streamValue$}	
	\If{$S.mode = matrix$}
		\State $mvToken = $ new $matrixValueToken$
		\State  $mvToken.value = v$
			\If{$S.right$ exists}
				\State send $mvToken$ to $S.right$
			\Else
				\State set $vectorFlag$ to right
				\State when particle stops at flag such that $S.right$ exists, pass $mvToken$ to $S.right$
			\EndIf
	\EndIf
	\If{$S.mode = vector$}
		\State $vvToken = $ new $vectorValueToken$
		\State  $vvToken.value = v$
		\State Pass $vvToken$ to $S.right$
	\EndIf
	\EndCase
	\Case {$endOfColumn$}
		\State $eocToken =$ new $end OfColumnToken$
		\State send $eocToken$ to $S.right$
	\EndCase
	\Case{$endOfMatrix$}
		\State $mode = vector$
	\EndCase
	\Case{ $endOfVector$}
		\State $mode = none$
		\State $sToken =$ new $startMultiplicationToken$
		\State send $sToken$ to $S.right$
	\EndCase
\EndSwitch
\EndProcedure
\end{algorithmic}
\end{algorithm}

 \begin{algorithm}
\caption{Matrix Particle M}\label{matrixparticlesetup}
\begin{algorithmic}[H]
\State $M.down =$ neighbor below (down column)
\State $M.value=$ value of matrix at location
\State $M.right =$ neighbor to right (across row)
\Procedure{MatrixParticleSetup}{}
\State $M.value = null$
\If{receive matrixValueToken $t$ with value $m$}
\If{$M.value = null$}
\State $M.value = m$
\Else
\If{M.down exists}
\State pass $t$ to $M.down$
\Else
\State set $matrixFlag$ down
\State when particle stops at flag such that $M.down$ exists, pass $t$ to $M.down$
\EndIf
\EndIf
\EndIf
\If{receive $vectorValueToken$ with value $v$}
\State create v instances of $productToken$
\While {$M.right$ has available space and $M$ has a $productToken$}
\State pass a $productToken$ to $M.right$
\EndWhile
\EndIf
\EndProcedure
\end{algorithmic}
\end{algorithm}

\begin{algorithm}
\caption{Vector Particle V}\label{vectorparticlesetup}
\begin{algorithmic}[H]
\State $V.right =$ neighbor to right (across row)
\State $V.right' =$ neighbor to right plus one rotation clockwise
\State $V.down =$ neighbor below (down column)
\State $V.value=$ value of vector at location
\State $V.columnFinished =$ when true, send values across to other columns instead of down
\Procedure{VectorParticleSetup}{}
\State $V.value = null$
\Switch{Type of token $t$ received}
	\Case{$vectorValueToken$ with value $v$}
		\If{$V.value = null$}
			\State $V.value = v$
			\State pass $t$ to $V.Down$
			\State create $resultCounterToken$ and pass it to $V.right$
		\Else
			\State pass $t$ to $V.Right$
		\EndIf

	\EndCase
	\Case{$matrixValueToken$}
		\If{$V.columnFinished = false$}
			\If{V.down exists}
				\State pass $t$ to $V.down$
			\Else
			\State set $matrixFlag$ downward
			\State when particle stops at flag such that $V.down$ exists, pass $t$ to $V.down$
			\EndIf
		\Else
			\If{$V.right$ exists}
				\State pass $t$ to $V.right$
			\Else
				\State set $vectorFlag$ to right
				\State when particle stops at flag such that $V.right$ exists, pass $t$ to $V.right$
			\EndIf
		\EndIf

	\EndCase
	\Case{$resultCounterToken$}
			\If{$V.right$ exists}
				\State pass $t$ to $V.right$
		\ElsIf{$V.right'$ exists}
				\State pass $t$ to $V.right'$
			\Else
				\State set $resultFlag$ to $right+1$ direction
				\State when particle stops at flag such that $V.right'$ exists, pass $t$ to $V.right'$
			\EndIf

	\EndCase
	\Case{$EndOfColumnToken$}
		\If{$V.columnFinished = false$}
				\State $V.columnFinished = true$
		\Else
			\State pass $t$ to $V.right$
		\EndIf
	\EndCase
\EndSwitch
\EndProcedure
\end{algorithmic}
\end{algorithm}

 \begin{algorithm}
\caption{Result Particle R}\label{resultparticlesetup}
\begin{algorithmic}[H]
\State $R.value = 0$
\State $R.right =$ neighbor to right (across row)
\State $R.down =$ neighbor below (down column)
\State $maxValue = $ maximum value a particle can hold before a carryover is needed
\State $resultRowCounterValue= $ used for $resultCounterTokens$ to recruit $\log w$ result columns
\Procedure{ResultParticle}{}
\If{receive $totalToken$}
\State $R.value ++$
\EndIf
\If{$R.value = maxValue$ }
\State  $tToken =$ new $totalToken$
\If{$R.right$ does not exist}
\State set $resultsFlag$ to right
\State when particle stops at flag such that $R.right$ exists, pass $tToken$ to $R.right$
\ElsIf{$R.right.value<maxValue$}
\State pass $totalToken$ to  $R.right$
\State $R.value = 0$
\EndIf
\EndIf
\If{Receive  $resultCounterToken$}
\State $resultRowCounterValue ++$
\If{$resultRowCounterValue = maxValue$}
\State $R.resultRowCounterValue = 0$
\State pass a new $resultCounterToken$ to $R.right$
\EndIf
\EndIf
\EndProcedure
\end{algorithmic} 
\end{algorithm}

\clearpage
\subsection{Proofs of Section~\ref{sec:analysisbinaryparticle}}
\label{appendix:counter}
\begin{proof}
\textbf{[Lemma~\ref{asyncdominates}]}
We prove by induction on rounds. Let $(C_0',C_1',...C_f')$ be an asynchronous movement schedule (Definition~\ref{asyncdef}) and let $(C_0,C_1,...C_f)$ be a parallel movement schedule such that $C_0 = C_0'$ (Definition~\ref{syncdef}).
For the base case, both schedules start in the same configuration, i.e. $C_0 = C_0'$, so the parallel version is dominated. 
Suppose $C_i$ is the configuration after round $i$ of the parallel schedule and $C_i'$ is the configuration after $i$ rounds of the asynchronous schedule, and assume $C_{i-1}\preceq_B C_{i-1}'$. For $C_{i}\succ C_{i}'$ to happen, there would have to be a carryover that moves forward in the parallel schedule but does not move in the asynchronous schedule. Consider a carryover $t$ with $p_{C_{i-1}}(t) = j$.\\\\
Case 1: Suppose $p_{C_{i-1}'}(t) \neq j$. Since $C_{i-1}'\succeq_B C_{i-1}$, we have $p_{C_{i-1}'}(t) >j$. Since $t$ can move at most one position forward in the parallel schedule in one round, $p_{C_i'}(t)\geq p_{C_i}(t)$, and $C_i\preceq_B C_i'$. \\\\
Case 2: Otherwise, suppose $p_{C_{i-1}'}(t) = j$. Then $t$ will only move forward if another carryover is handed to $t$'s current particle. If $P_j$ does not receive such a new carryover in $C_i'$, it must be that $P_{j-1}$ was not holding two carryovers in $C_{i-1}'$, so $P_{j-1}$ is not currently executing $sendToken()$. Then since $C_{i-1}\succeq_B C_{i-1}'$, $P_{j-1}$ can not have two carryovers in $C_{i-1}$, because if it had two, one would be farther forward than in $C_{i-1}'$. Thus $P_{j-1}$ is also not currently executing $sendToken()$. Then there is no way $P_j$ can receive a new carryover, so the state of $P_j$ remains the same in $C_i$ and $C_i'$ from the previous round, so $C_i \preceq_B C_i'$. \\\\
 Note that it is also possible for the asynchronous schedule to make greater progress than the parallel version. If particle $P_{N+1}$ is not activated until after the particles $P_0,P_1,...P_{N}$ are activated in the same order repeatedly, a carryover can move up to $N$ distance in a single asynchronous round.
\end{proof}
\begin{proof}
\textbf{[Lemma~\ref{counter:syncOn}]}
Since a displayed bit is flipped ($P_j.display$ is changed for some $j$) every time a carryover is moved, if a carryover moves at every round where $C_i\neq C_f$, we can apply amortized analysis in terms of displayed bit flips only.\\ 
Consider a configuration $C_i \neq C_f$. Then of the carryovers not in their final position, there is token $t$ held by the particle with highest index; let $p_{C_i}(t) = j$. Then consider $P_{j+1}$: since every carryover held by a particle $P_k$ with $k>j$ is in its final position, there cannot be two carryover tokens held by $P_{j+1}$, so in configuration $C_{i+1}$, particle $P_{j+1}$ can accept $t$; thus a carryover has moved and a bit has flipped from $C_i$ to $C_{i+1}$. \\
$P_0$ flips its displayed bit once per each parallel round, $P_1$ flips every other parallel round; for $j=  0,1...\lceil \log v\rceil$, the particle $P_j$ flips the displayed bit $\lceil \frac{v}{2^j}\rceil$ times when counting to the value $v$ (i.e. processing $v$ carryovers to their final positions), and $P_j$ for $j>\lceil \log v\rceil$ does not ever flip. Then the total number of flips, and thus carryover movements, is:
$$\displaystyle\sum_{j=0}^{\lfloor \log v\rfloor} \lfloor \frac{v}{2^j}\rfloor < v \sum_{j=0}^\infty \frac{1}{2^j} = 2v$$
 so the algorithm runs in $O(v)$ parallel rounds.\\
\end{proof}

\begin{proof}
\textbf{[Theorem~\ref{counter:asyncOn}]}
By Lemma \ref{asyncdominates}, after any given number of rounds the asynchronous counter will be at least as close to finished as the parallel counter. By Lemma \ref{counter:syncOn} the parallel counter counts to $v$ in $O(v)$ rounds and so the asynchronous counter does as well. \\
The counter receives $v$ carryovers originating at the seed particle, $S$. To reach the final configuration $C_f$, each carryover must be created and handed off to particle $P_0$. Since $S$ and $P_0$ can perform at most one token handoff per round, the $O(v)$ bound cannot be improved and the asynchronous counter completes in $\Theta(v)$ rounds.
\end{proof}

\subsection{Proofs of Section\ref{sec:analysismatrixmult}}
\label{appendix:matrixmult}
\begin{proof}
\textbf{[Theorem~\ref{streamingcorrect}]}

It can be shown by induction that each value reaches the correct position for the system to be executed. For the base case, consider the first value the seed receives, $m_{0,0}$. Following Algorithm~\ref{seedparticlesetup}, the seed particle sets the $vectorFlag$ in the direction toward where the vector will be built. Since free particles move along the perimeter of the structure, a free particle is eventually in the position pointed to by the seed's $vectorFlag$, and thus stops and becomes a vector particle. Upon receiving the matrix value $m_{0,0}$, by Algorithm~\ref{vectorparticlesetup}, the first vector particle will set the $matrixFlag$ pointing in the direction to build the column (determined relative to the position of the seed), and a free particle will eventually stop at position $M_{0,0}$. When the vector particle passes the value to that matrix particle, $m_{0,0}$ has reached its correct position, $M_{0,0}$.\\ 
Consider a matrix value $m_{u,v}$ (with $(u,v)\neq (0,0)$) somewhere in the stream, so it should be in row $u$, column $v$ of the final $h \times w$ matrix, counting from the corner at which that the seed initiates the construction. Then if the stream is correct, $m_{u,v}$ appears in the stream after exactly $v-1$ end of column markers and is the $u+1^{th}$ value after the most recent column marker, $e_{v-1}$. Suppose all the values up to $m_{u,v}$ have been placed correctly in particles, so there are $v$ completed columns of $h$ particles and $u$ particles in the $v+1^{th}$ column. There are two cases for how value $m_{u,v}$ is handled.\\\\
Case 1: Suppose $i=0$, so $m_{u,v}$ corresponds to the first particle in a new column. Then $v$ $EndOfColumn$ tokens have entered the system, so the first $v$ vector particles have set $columnFinished$ to true and the $v^{th}$ vector particle has recruited vector particle $V_{v}$. Then $m_{u,v}$ will be passed by each vector particle until it reaches $V_{v}$, and since $V_{v}$ does not have $columnFinished$ set to true, it will set the value of $V_v.matrixFlag$ to point to where the column should be constructed, as stated in Algorithm~\ref{vectorparticlesetup}. As in the base case, a free particle will reach the position pointed to by  $V_v.matrixFlag$ and become a matrix particle. $V_v$ will then pass $m_{u,v}$ to the new matrix particle, $M_{u,v}$, so $m_{u,v}$ has reached the correct final position.\\\\
Case 2: Otherwise suppose $u>0$. Since the first $v-1$ vector particles have set $columnFinished$ to true they will pass $m_{u,v}$ to the next vector particle to the right, so it will reach the $v^{th}$ vector particle, $V_{v}$. Since $V_v$ has not set $columnFinished$ to true, it will pass $m_{u,v}$ down the column, as stated in Algorithm~\ref{vectorparticlesetup}. Each of the first $u-2$ matrix particles have a neighbor below themselves, so they will also pass $m_{u,v}$ down the column, so that $M_{u-1,v}$ will hold value $m_{u,v}$. Following Algorithm~\ref{matrixparticlesetup}, $M_{u-1,v}$ will set $M_{u-1,v}.matrixFlag$ to point downward (i.e. away from $M_{u-2,v}$). As in the base cases, eventually there will be a free particle in the position pointed to by $M_{u-1,v}.matrixFlag$, and that free particle will become matrix particle $M_{u,v}$. Then $M_{u-1,v}$ will pass $m_{u,v}$ to $M_{u,v}$, and so the value will be at row $u$ and column $v$ of the matrix. Since no matrix value moves after reaching its correct position, the final configuration is correct.
\end{proof}

\begin{proof}
\textbf{[Lemma~\ref{setupdomination}]}

 Let $(C_0',C_1',...C_f')$ be an asynchronous movement schedule (Definition~\ref{asyncdef}). We show a  parallel movement schedule $(C_0,C_1,...C_f)$ such that $C_0 = C_0'$ (Definition~\ref{syncdef}) can be selected such that $C_i'\preceq C_i'$ for all $i$ such that $0\leq i \leq f$.
 
First, consider particle positions. Free particles follow the spanning forest primitive defined in Section~\ref{spanningforestappendix}. If we consider the structure as a tree rooted at the seed, the parallel schedule $(C_0,C_1,...C_f)$ is a forest schedule. Furthermore, each particle which will eventually become part of the structure follows a predefined path and moves whenever possible, so the parallel schedule $(C_0,C_1,...C_f)$ is a \emph{greedy forest schedule} as defined in~\cite{DBLP:conf/dna/DerakhshandehGP16}. Then we can apply Lemma 3 from~\cite{DBLP:conf/dna/DerakhshandehGP16} which claims that a parallel schedule $(C_0,C-1,...C_f)$ with $C_0 = C_0'$ can always be chosen such that $C_i\preceq C_i'$ for all $i$ in terms of particle positions.

 Next, we induct on rounds to show that a parallel schedule selected based on particle positions is also dominated by the asynchronous schedule in terms of token positions.  For the base case, both schedules start in the same configuration, so the asynchronous schedule dominates the parallel version. Then let $C_i$ be the configuration of the parallel version after $i$ rounds and $C_i'$ be the configuration of the asynchronous version after $i$ rounds. Suppose $C_{i-1}\preceq_M C_{i-1}'$. We show it is not possible that $C_i\succ_M C_i'$. To proceed, we can assume $C_{i-1} = C_{i-1}'$ because if $C_{i-1}\prec_M C_{i-1}'$, some amount of progress can be removed from the asynchronous execution to produce $C_{i-1}=C_{i-1}$. In the parallel schedule a token will move forward at most by one position, if the position at $p_{C_i}(t)+1$ becomes open during the round $i$. This means that if there is a sequence of tokens in adjacent particles and there is an empty particle on the end in the direction of token motion, all of the tokens will move forward by one position, as described in Definition~\ref{tokenpathdef}. In the asynchronous schedule, we preserve this ability to guarantee forward motion of tokens by allowing the capacity of each particle to be two tokens. At the beginning of round $i$, if $C_{i-1} = C_{i-1}'$, each particle has at most one token (since token capacity in the parallel configuration is one). Then, in any order, every token can move forward by one position as long as the particle holding it activates, which is guaranteed in the course of a single asynchronous round.  Additionally, since tokens do not follow paths that branch and merge, execution order in an asynchronous round does not affect token paths or ordering (when holding two tokens a particle preserves order). Thus when the particle at $p_{C_i}(t)$ is activated, the token $t$ will at least move forward one spot regardless, and in no case will a token move in the asynchronous schedule and not in the parallel schedule to produce $C_i \succ C_i'$.\\\\

\end{proof}

\begin{proof}
\textbf{[Theorem~\ref{matrixOnp}]}

 We assume there are exactly $n'$ particles (other than the seed) in the system. First suppose particles are present in every final position of the system from the first round, so that when a particle sets a flag to create an additional matrix or vector particle, a free particle is already in position to become the new matrix or vector particle. Then at each construction position it takes $O(1)$ rounds to recruit a new particle because the previous particle (in the structure) sets a flag and then the new particle, already in position, changes state. Pipelining the tokens in the parallel version as in Lemma~\ref{setupdomination} means that after a new particle is recruited it again takes only $O(1)$ rounds for it to receive its position value and the position value for the next position, which it is then responsible for recruiting. Thus the total setup is $O(1)$ for each of the $O(n')$ matrix particles. 
 
 Then for the set of vector values, the total time for the values to each reach their corresponding vector particle is equal to the time for the last vector value to reach the particle farthest from the seed. Thus it takes $O(w)$ rounds for vector values to fill into the particles already in place for the vector. 
 
 Finally, as each vector particle receives a vector value, they create and send a result counter token. These tokens are used to determine an upper bound on how many particles to recruit to store the results. The particles on the end of what will be the results block (the top row of blue particles in Figure~\ref{setupfig}) act as a counter, using result counter tokens, to initiate construction of $\log_c{(m^2w)}$ columns. The result counter tokens move simultaneously with the vector values, so when the last vector value gets to the farthest vector particle, the result counter tokens have already started to move into the end row of result particles. A new result counter token will then arrive at $R_{0,0}$, and thus enter the counter, at each parallel round, so it will take $O(w)$ rounds to finish counting as in Theorem~\ref{counter:syncOn}. 
 Thus the total number of rounds to finish setup when no particle motion is needed is $O(n')+O(w)+O(w) = O(n')$.
 
Then consider a system using the particle stopping algorithm from an initial spanning forest configuration, so that not all final positions are immediately occupied with free particles. The initial spanning forest setup requires $O(n')$ parallel rounds. Then for column $v$, suppose column $v-1$ (if it exists) becomes completely finished, particles fill column $v$ (in $prestop$ state), and then column $v$ executes (i.e. begins passing tokens and changing state to Matrix or Result). In reality these phases overlap. 

Each column takes $O(h)$ rounds to fill with particles after previous column finishes executing token passes, so $O(hw) = O(n') $ for all column filling is added to the total time for token passing. This additional time is added by pausing and resuming the token passing to insert periods of column filling. The total is $O(n')$. 
\end{proof}
\begin{proof}
\textbf{[Theorem~\ref{matrixthetanp}]}

The parallel setup configuration completes in $O(n')$. Since at every round $i$, the parallel configuration is dominated by the asynchronous configuration, if $f$ is the final round of parallel execution, $C_f \preceq_M C_f'$, so the asynchronous execution is also complete by round $f$ and thus asynchronous construction is $O(n')$.\\
Consider the role of the seed particle in setting up the system. It must pass each position token (for matrix and vector particles) into the system via $V_0$ individually. Passing at most one token per round, it takes at least $hw$ rounds for all positions to enter the system, and thus setup is $\Omega(n')$.
\end{proof}
\begin{lemma}\label{flattenlemma}
In parallel execution, all trees flatten into a line of particles to supply construction after $V_0$. This line never has a gap at the current location of construction (i.e. current position requiring a particle to stop and join the system).
\end{lemma}
\begin{proof}
\textbf{[Lemma~\ref{flattenlemma}]}

We induct on rounds. For the base case, the system must be initially connected, and all particles are contracted. Thus the seed particle has at least one neighbor that will, as a result, become a leader. An initial neighbor of the seed will become $V_0$ and expand. Followers then adjacent to $V_0$ will initially be contracted. Since we are considering a greedy forest schedule, Lemma 2 from~\cite{DBLP:conf/dna/DerakhshandehGP16} applies and we have that every expanded parent in every $C_i$ for $1\leq i \leq f$ has at least one contracted child. From this point, leader particles will try to move as far counterclockwise around the existing system as possible, until reaching the end of the vector line, meaning they will not circle indefinitely. In the worst case, the first particle $P_0$, which will become $M_{0,0}$, is the first follower particle whose leader became $V_0$ (Figure~\ref{setupcol}a.). In this case there must be no particles moving around the structure in front of $P_0$, or they would be recruited to the position first. Then since we assume there are sufficiently many particles to build the system, at least $n'-1$ particles remain in either leader or follower state, so they are each part of spanning trees rooted at $P_0$ or part of trees rooted behind $P_0$. 

 Consider when any leader particle $P_j$ is expanded and preparing to contract. If there is another leader immediately behind $P_j$, it may perform a handover with $P_j$. If no such leader exists,  $P_j$ has followers, one of which will perform a handover and become a leader. Thus in any case there is a particle filling the position immediately behind $P_j$. This will continue to happen at every contraction of $P_j$ and all other leaders until follower particles are all leaders and the structure is complete. In any other starting configuration, particles in trees or on the structure in front of this initial position are ahead of the worst case.
\end{proof}

\begin{proof}
\textbf{[Lemma~\ref{linearcolfill}]}

By using the $prestop$ particle state to stop particles in the correct position regardless of how much progress the token passing for recruiting vector and matrix particles has made, we can consider the motions for column filling as completely independent of token passing for matrix and vector values.\\\\
Column $0$:  Initially, all particles are in trees rooted at the seed. We induct on $u$ where the particles of the first column are $\{M_{u,0}|0\leq u<h\}$. In the worst case described in Lemma~\ref{flattenlemma}, the particle $P_0$ immediately becomes a leader and travels $7$ units (taking $15$ parallel rounds) to reach the position for $M_{0,0}$ (Figure~\ref{setupcol}a.). $P_1$ is also a special case, taking only $2$ additional rounds so that $P_0$ and $P_1$ are in place after $17$ parallel rounds. 

Now suppose $P_k$ has finished as $M_{k,0}$ and consider $P_{k+1}$.  While $P_k$ is moving, in the worst case there is another follower particle which moves downward behind $P_k$, call it $P_{k+1}$. Since we assume there are sufficient particles in the system and in the worst case no trees were rooted closer to the construction site than $P_0$, $P_{k+1}$ must be a follower of $P_k$, or have become one in the processing of the spanning trees before reaching the structure. Then in parallel execution, $P_{k+1}$ remains at most one unit behind $P_{k}$, meaning it must travel two units after $P_k$ finishes. This requires $5$ rounds, so suppose $P_k$ finished in $17+5(k-1)$ asynchronous rounds. Since $P_{k+1}$ is at most $1$ space behind $P_k$, $P_{k+1}$ will be finished after $17+5k$ parallel rounds. Thus the $P_{h-1}$ finishes within $17+5(h-1)$ rounds, so the first column finishes in $O(h)$ parallel rounds.\\\\
Column $1$: When the last particle of column $0$ changes state to $prestop$ or $matrix$, the next leader, $P_{x}$, which will become $M_{0,1}$, is the neighbor of $M_{h-1,0}$ immediately counterclockwise from column $0$ (Figure~\ref{setupcol}b.). In this proof we ignore the time from when the last particle of column $0$ was in position until when it changed state since we are only considering the execution time of particle motions. It then takes $6$ rounds for $P_h$ to reach the position directly below that which will eventually become $M_{h-1,1}$. From there $2$ rounds are required to move each of the $h$ positions, so in total, $P_h$ takes $2h-6$ rounds to reach its final position and become $M_{0,1}$. Each subsequent particle in column $1$ follows $P_h$ in a connected line, so each additional particle requires only $1$ rounds to finish, so column $1$ is constructed in a total of $O(h)$ rounds.\\\\
Column $v$, $2\leq v <w$: Each subsequent column begins to fill when the previous row is completely occupied with $prestop$ or $matrix$ particles, so the first leader of these columns is always in the same relative position at the start of construction (Figure~\ref{setupcol}c.). It takes $2$ rounds for the particle to reach the relative starting position of the first leader constructing column $1$, and from there each column construction is the same and requires $O(h)$ parallel rounds.
\end{proof}
\begin{proof}
\textbf{[Lemma~\ref{matrix:syncOn}]}
We have shown that setup of the entire matrix-vector multiplication structure, including result particles, requires $O(n')$ rounds, so we now assume no further particle motion is necessary and consider only token motion for the actual multiplication. Since a token moves forward in parallel each round, the vector tokens take exactly $w$ rounds to move across the matrix from their start to finish locations.  For simplicity, a value $z$ resulting from a multiplication at a matrix particle is represented as $z$ copies of a product token, which can be passed individually to particles with unused storage capacity.  Then each product token originating at $M_{u,v}$ takes $w-v$ rounds to move from its creation location to its destination. The product tokens starting in row $h-1$ are created last, and the ones that start at $M_{h-1,0}$ have the farthest distance to reach the end of the matrix. Since each vector and matrix value is less than a constant maximum $m$, the maximum number of product tokens produced in a row of the matrix is $m^2w$. Then the farthest possible result particle a sum token may travel to is $\log( m^2w) = 2\log m + \log w =O(\log w)$, which is in addition to moving across at most $w$ matrix particles. Then product tokens reach their final position in $O(h+w+\log w) = O(h+w)$ rounds.
\end{proof}
\begin{proof}
\textbf{[Theorem~\ref{sumthetahw}]}
By Lemma \ref{setupdomination}, after any given number of rounds the asynchronous multiplier will be at least as close to finished as the parallel multiplier. By Lemma \ref{matrix:syncOn} the parallel multiplier finishes in $O(h+w)$ rounds and so the asynchronous multiplier does as well. \\
In the asynchronous execution, it is possible that vector tokens require only one round to go from their start to finish locations. However, each of the $w$ product tokens for a matrix row must pass through the first results particle for that row, and that particle can only process one product token per round. Thus the overall matrix-vector multiplication is $\Omega(h+w)$. 
\end{proof}

\begin{proof}
\textbf{[Theorem~\ref{sumthetahwy}]}
Consider matrix-matrix multiplication as a sequence of $y$ matrix-vector multiplications. The first matrix-vector multiplication completes in $\Theta(h+w)$ rounds, by Theorem~\ref{sumthetahw}. Each subsequent vector must be set up by streaming $w$ vector values into the system and moving one to each vector particle. This requires $O(w)$ parallel rounds (and thus $O(w)$ asynchronous rounds) since each vector value must be passed at most $w$ times. The seed must be activated at least once for each of the $w$ new vector values so the asynchronous system requires $\Theta(w)$ rounds to setup the new vector. Thus the total matrix-vector multiplication is $\Theta(y(h+w))$.
\end{proof}

\subsection{Proofs of Section~\ref{sec:edgedetect}}
\begin{proof}
\textbf{[Theorem~\ref{edgeconst}]}
Assume all particles are holding their pixel value, meaning setup has finished completely. In two asynchronous rounds, gathering surrounding raw pixel data and completing a convolution for noise filtering can be finished by every particle. In one round each particle will gather the values of its neighbors and in a second complete round each particle will gather the values of particles at a distance of two and complete the convolution operation. Gathering updated values and performing a second convolution for Sobel filtering (Algorithm~\ref{pixelcanny}) similarly requires only two single asynchronous rounds for the results of the first operation to propagate a distance of two hops. Finally, establishing edge certainty can be accomplished in as many asynchronous rounds as there are certainty levels. In the first round, very sure particles will establish themselves and in each subsequent round particles with sure neighbors will identify themselves.
\end{proof}

\end{document}